\pdfoutput=1

\documentclass[final,twoside,11pt]{entics}
\usepackage{preamble}
\usepackage{enticsmacro}
\usepackage{graphicx}
\usepackage{iblock}

\def\conf{MFPS 2023}
\def\myconf{mfps39}
\volume{3}
\def\volu{3}
\def\papno{1}
\def\mydoi{12232}
\def\copynames{F. Aagaard, J. Sterling, L. Birkedal} 
\def\lastname{Aagaard, Sterling, and Birkedal}
\def\runtitle{A denotationally-based program logic for higher-order store}
\def\copynames{F. Aagaard, J. Sterling, L. Birkedal}

\begin{document}

\begin{frontmatter}

  \title{A Denotationally-based Program Logic\\ for Higher-order Store}

  \author{Frederik Lerbjerg Aagaard\thanksref{freemail}}
  \author{Jonathan Sterling\thanksref{jonemail}}
  \author{Lars Birkedal\thanksref{laremail}}

  \address{Department of Computer Science\\ Aarhus University\\ Aarhus, Denmark}

  \thanks[freemail]{Email: \href{mailto:aagaard@cs.au.dk}
    {\texttt{\normalshape aagaard@cs.au.dk}}}
  \thanks[jonemail]{Email: \href{mailto:jsterling@cs.au.dk}
    {\texttt{\normalshape jsterling@cs.au.dk}}}
  \thanks[laremail]{Email: \href{mailto:birkedal@cs.au.dk}
    {\texttt{\normalshape birkedal@cs.au.dk}}}

  \begin{abstract}
    Separation logic is used to reason locally about stateful programs. State
    of the art program logics for higher-order store are usually built on top
    of untyped operational semantics, in part because traditional denotational
    methods have struggled to simultaneously account for general references and
    parametric polymorphism. The recent discovery of simple denotational
    semantics for general references and polymorphism in synthetic guarded
    domain theory has enabled us to develop \LMuRef, a higher-order separation
    logic over the \emph{typed equational theory} of higher-order store for a
    monadic version of System~\FMuRef{}.
    The \LMuRef{} logic differs from operationally-based program logics in two
    ways: predicates range over the \emph{meanings} of typed terms in an
    arbitrary model rather than over the \emph{raw code} of untyped terms, and
    they are automatically invariant under the equational congruence of
    higher-order store, which applies even underneath a binder. As a result,
    ``pure'' proof steps that conventionally require focusing the Hoare triple
    on an operational redex are replaced by a simple equational rewrite in
    \LMuRef.  We have evaluated \LMuRef{} against standard examples involving
    linked lists in the heap, comparing our abstract equational reasoning with
    more familiar operational-style reasoning. Our main result is the soundness
    of \LMuRef{}, which we establish by constructing a BI-hyperdoctrine over
    the denotational semantics of \FMuRef{} in an impredicative version of
    synthetic guarded domain theory.
  \end{abstract}

\end{frontmatter}

%
%
%
\section{Introduction}

State of the art program logics such as Iris~\cite{iris:2018} and the Verified
Software Toolchain~\cite{appel:2011} typically combine two design decisions:
the basic sorts of the logic are generated by the \emph{raw, untyped terms} of
the programming language, and constructs of the program logic are defined in
terms of operational semantics layered over this raw syntax. In contrast, the
original program logic
LCF~\cite{gordon-milner-wadsworth:1979,plotkin:1977,paulson:1987,gordon:2000}
and successors like HOLCF~\cite{huffman:2012} work in a more abstract way:
predicates range over \emph{typed denotations} rather than over untyped raw
syntax. The difference in abstraction is substantial: whereas a predicate on
raw syntax need not even be invariant under $\alpha$-equivalence, a predicate
on denotations is automatically invariant under the \emph{entire} equational
theory of the language. Nevertheless, denotationally-based program logics have
fallen by the wayside in part because of the historic failure of denotational
methods to provide simple enough answers to the questions posed by higher-order
effectful programming, \eg the combination of general references, polymorphism,
and concurrency; as a result, the benefits of denotations are primarily reaped
today in the first-order realm, as in the successful revival of coinductive
resumption semantics by Xia~\etal~\cite{itrees:popl:2019}.

Denotational and algebraic approaches to program semantics have other benefits
besides their abstractness: they are more modular, and they are more directly
compatible with practical mathematical tools from category theory and topology,
which have become increasingly relevant in light of recent interest in
probabilistic and differentiable computing for the sciences. We therefore have
ample motivation to bridge the gap between denotational methods and modern
logics for reasoning about programs.

\subsection{Denotational semantics of general references and polymorphism}

Sterling \etal~\cite{sterling-gratzer-birkedal:2022} have recently discovered a simple
denotational semantics for general references and polymorphism, decomposing the
problem of higher-order store with semantic worlds into two separate constructs
that are easily combined: impredicativity and guarded recursion. At the heart
of their development is \DefEmph{Impredicative Guarded Dependent Type Theory}
(\iGDTT{}), a type theoretic metalanguage for synthetic guarded domain
theory~\cite{bmss:2011} extended by an impredicative universe of sets.
Although \opcit have made promising strides in the denotational understanding
of higher-order stateful programs, non-trivial \emph{reasoning} about state
requires more than equational logic.  Operationally-based program logics like
Iris employ \DefEmph{guarded higher-order separation logic} for this purpose; our
contribution in this paper is to adapt these methods to the denotational semantics of
Sterling~\etal.

\subsection{Higher-order separation logic over denotational semantics}

We extend the \emph{monadic equational theory} of higher-order store for
System~\FMuRef{} to a full higher-order separation logic called \LMuRef. In
contrast to systems like Iris which layer the separation logic over an
encoding of the \emph{raw syntax} of the object language, predicates in
\LMuRef{} range not over code but rather over the actual types of \FMuRef{},
subject to the full equational congruence.  In a conventional
operationally-based program logic, equational reasoning is rather only applicable
within the focus of a Hoare triple and is moreover limited to operational head
reduction and cannot occur under binders, \etc.

Our main contribution is the soundness of \LMuRef{} for reasoning about
Sterling~\etal's denotational semantics of \FMuRef{}, which we establish by
constructing a suitable
BI-hyperdoctrine~\cite{biering-birkedal-torp-smith:2007}. We have made
pervasive use of type theoretic internal languages to facilitate this work:
first, System~\FMuRef{} is intepreted in \emph{internal co-presheaves} on a
poset of semantic heap layouts $\Worlds$ defined in \iGDTT{} by solving a guarded
recursive domain equation; then, our BI-hyperdoctrine is itself constructed in
the type theoretic language of co-presheaves on $\Worlds$. Because \iGDTT{}
itself is sound~\cite{sterling-gratzer-birkedal:2022}, it follows that our logic is sound.

\subsection{Structure of this paper}

In \cref{sec:L-mu-ref} we describe the \LMuRef{} logic and its computational
substrate \FMuRef{}; in \cref{sec:example}, we illustrate the use of \LMuRef{}
through a worked example involving linked lists in the heap. In
\cref{sec:igdtt}, we recall the \iGDTT{} metalanguage and the basics of
synthetic guarded domain theory, which we use in
\cref{sec:semantics-of-store,sec:semantics-of-predicates} to explain the
denotational semantics of higher-order store and the interpretation of guarded
higher-order separation logic over this model. In \cref{sec:soundness}, we show
that \LMuRef{} is sound with respect to this model and therefore consistent.

\section{\LMuRef: an equational program logic for higher-order store}\label{sec:L-mu-ref}

In this section, we introduce a \LMuRef{} program logic that extends the
monadic equational logic of System~\FMuRef{} with the connectives
and rules of guarded higher-order separation
logic~\cite{bizjak-birkedal:2018}, together with a built-in connective for
\emph{weakest preconditions}. The \LMuRef{} logic is related to System~\FMuRef{} in roughly the
same way that \textbf{LCF} relates to \textbf{PCF}~\cite{plotkin:1977,scott:1993}.

\LMuRef{} will have two kinds of types --- \DefEmph{``program types''} and
\DefEmph{``logical types''} --- and to efficiently organize the formalism, we
will define the typehood judgment schematically in a flag
$\Hl{\iota\in\brc{\FlagP,\FlagL}}$.\footnote{Both System~F \emph{and}
higher-order logic are modeled by impredicative universes, but it is well-known
that one impredicative universe cannot contain another~\cite{coquand:1986},
which leads us to introduce the stratification of logical types from program types. The prime example of a logical
type that is not a program type will be the type of all propositions.} With
this convention in hand, the forms of
judgment of \LMuRef{} classify type contexts \mbox{\JdgBox{\IsTpCtx{\Xi}},}
element contexts \mbox{\JdgBox{\IsCtx{\Xi}{\Gamma}},} types
\mbox{\JdgBox{\IsTp{\iota}{\Xi}{A}},} program elements
\mbox{\JdgBox{\IsEl{\Xi}{\Gamma}{a}{A}},} propositions
\mbox{\JdgBox{\IsProp{\Xi}{\Gamma}{\phi}}}, and entailments
\mbox{\JdgBox{\Entails{\Xi}{\Gamma}{\phi}{\psi}}.}  Every syntactic construct
of \LMuRef{} is subject to a relation of \DefEmph{judgmental equality}
$\prn{\equiv}$ which is respected by all judgments and preserved by all
operations; for instance, we shall write $\Hl{\EqEl{\Xi}{\Gamma}{u}{v}{A}}$ to
mean that $u$ and $v$ are judgmentally equal elements of type $A$.

\begin{notation}
  When specifying the rules of \LMuRef{}, we will often write entailments
  $\Hl{\Entails{\Xi}{\Gamma}{\phi}{\psi}}$ as $\Hl{\BareEntails{\phi}{\psi}}$
  when the contexts do not change from premise to conclusion.
\end{notation}

\subsection{Rules for types and elements}

Type contexts $\Hl{\IsTpCtx{\Xi}}$ contain variables $\Hl{\alpha}$ ranging over
\emph{program types} only. Element contexts $\Hl{\IsCtx{\Xi}{\Gamma}}$ contain
variables $\Hl{x:A}$ ranging over elements of types
$\Hl{\IsTp{\iota}{\Xi}{A}}$. Both program types and logical types are
closed under function spaces, cartesian products, inductive types, and universal quantification over program types; program
types are furthermore closed under a computational monad in the sense of
Moggi~\cite{moggi:1991}, as well as existential and recursive
types. The only subtlety is that the $\texttt{unfold}$ destructor for recursive
types is treated \emph{effectfully}, confining recursion to the
monad. Every program type is a logical type, and the inclusion of program types
into logical types commutes up to isomorphism with function spaces, product
types, and inductive types. We will treat these coercions
silently as it causes no ambiguity.

For lack of space, we will not dwell on the product and function types except
to comment that they satisfy the full universal properties of exponentials and
cartesian products up to judgmental equality.  Likewise, we do not present here
the (quite standard) account of inductive types --- except to note that we
close both program and logical types under inductives, and structurally
recursive functions into logical types on program-level data are permitted.

\subsubsection{The computational monad}

The universe of program types is closed under a computational monad $\TpT$ in the
sense of Moggi~\cite{moggi:1991} governing both state and general recursion.
\begin{mathpar}
  \inferrule{
    \IsPTp{\Xi}{A}
  }{
    \IsPTp{\Xi}{\TpT\,A}
  }
  \and
  \inferrule{
    \IsEl{\Xi}{\Gamma}{u}{A}
  }{
    \IsEl{\Xi}{\Gamma}{\TmRet\,u}{\TpT\,A}
  }
  \and
  \inferrule{
    \IsEl{\Xi}{\Gamma}{u}{\TpT\,A}\\
    \IsEl{\Xi}{\Gamma,x:A}{v}{\TpT\,B}
  }{
    \IsEl{\Xi}{\Gamma}{x\leftarrow u; v}{\TpT\,B}
  }
\end{mathpar}

We present the standard equational theory of monads in
\cref{sec:omitted-rules}. We defer our discussion of the monadic operations
for general recursion and state to \cref{sec:effect-rules}.

\subsubsection{General recursion and general references}\label{sec:effect-rules}
Program types are closed under both general recursive types and general
reference types. As the rules for these are somewhat subtle in our monadic
environment, we cover them in detail. First we describe the formation rules for
types and terms, before introducing their equational theory.
\begin{mathpar}
  \inferrule{
    \IsPTp{\Xi,\alpha}{A}
  }{
    \IsPTp{\Xi}{\mu\alpha.A}
  }
  \and
  \inferrule{
    \IsEl{\Xi}{\Gamma}{u}{A\brk{\mu\alpha.A/\alpha}}
  }{
    \IsEl{\Xi}{\Gamma}{\TmFold\,u}{\mu\alpha.A}
  }
  \and
  \inferrule{
    \IsEl{\Xi}{\Gamma}{u}{\mu\alpha.A}
  }{
    \IsEl{\Xi}{\Gamma}{\TmUnfold\,u}{\TpT\,{A\brk{\mu\alpha.A/\alpha}}}
  }
  \\
  \inferrule{
    \IsPTp{\Xi}{A}
  }{
    \IsPTp{\Xi}{\TpRef\,A}
  }
  \and
  \inferrule{
    \IsEl{\Xi}{\Gamma}{u}{\TpRef\,A}
  }{
    \IsEl{\Xi}{\Gamma}{\TmGet\,u}{\TpT\,A}
  }
  \and
  \inferrule{
    \IsEl{\Xi}{\Gamma}{u}{\TpRef\,A}
    \\
    \IsEl{\Xi}{\Gamma}{v}{A}
  }{
    \IsEl{\Xi}{\Gamma}{\TmSet\,u\,v}{\TpT\,\prn{}}
  }
  \and
  \inferrule{
    \IsEl{\Xi}{\Gamma}{u}{A}
  }{
    \IsEl{\Xi}{\Gamma}{\TmNew\,u}{\TpT\,\prn{\TpRef\,A}}
  }
  \and
  \Alert{
    \inferrule{
    }{
      \IsEl{\Xi}{\Gamma}{\TmStep}{\TpT\,\prn{}}
    }
   }
\end{mathpar}

The $\TmStep$ constructor above is a ``no-op'' instruction witnessing the
\emph{guarded} or \emph{intensional} nature of our denotational
semantics;\footnote{In the presence of polymorphism, the guarded interpretation
of state and general recursion seems to be
forced~\cite{birkedal-stovring-thamsborg:2010}.} this is reflected in the
equational theory by the emission of $\TmStep$ instructions with reads to the
heap and the unfolding of recursive types.  Our logic, as we explain later,
will provide rules that allow these steps to be used as fuel for recursive
deductions.
\begin{mathpar}
  \inferrule[unfold-of-fold]{}{
    \BareEqEl{
      \TmUnfold\,\prn{\TmFold\,u}
    }{
      \Alert{\TmStep}; \TmRet\,{u}
    }{
      \TpT\,A\brk{\mu\alpha.A/\alpha}
    }
  }
  \and
  \inferrule[fold-of-unfold]{}{
    \BareEqEl{
      {x\leftarrow \TmUnfold\,u; \TmRet\,\prn{\TmFold\,x}}
    }{
      \Alert{\TmStep};\TmRet\,u
    }{
      \TpT\,\prn{\mu\alpha.A}
    }
  }
  \and
  \inferrule[get-after-set]{}{
    \BareEqEl{
      \TmSet\,u\,v; \TmGet\,u
    }{
      \Alert{\TmStep}; \TmSet\,u\,v; \TmRet\,v
    }{\TpT\,A}
  }
  \and
  \inferrule[set-after-new]{}{
    \BareEqEl{
      \prn{x\leftarrow\TmNew\,u;
      \TmSet\,x\,v;
      w}
    }{
      \prn{x\leftarrow\TmNew\,v;
      w}
    }{\TpT\,B}
  }
  \and
  \inferrule[set-after-set]{}{
    \BareEqEl{
      \TmSet\,u\,v;
      \TmSet\,u\,w
    }{
      \TmSet\,u\,w
    }{\TpT\,A}
  }
  \and
  \inferrule[get-after-get]{}{
    \BareEqEl{
      \prn{x\leftarrow\TmGet\,u;
      y\leftarrow\TmGet\,v;
      w}
    }{
      \prn{y\leftarrow\TmGet\,v;
      x\leftarrow\TmGet\,u;
      w}
    }{\TpT\,C}
  }
  \and
  \inferrule[set-after-get]{}{
    \BareEqEl{
      \prn{
        x\leftarrow\TmGet\,u;
        \TmSet\,u\,x;
        w
      }
    }{
      \prn{
        x\leftarrow\TmGet\,u; w
      }
    }{\TpT\,B}
  }
\end{mathpar}

We also have, but do not present, rules that commute $\TmStep$ past all
primitive effects.

\begin{remark}
  Observe that ours is a theory of \emph{higher-order global
  store} rather than \emph{higher-order local store} in the sense that
  allocations are not hidden but rather have a globally observable effect. The state of the art in denotational
  semantics for local store is currently restricted to references of \emph{ground
  type}~\cite{kammar-levy-moss-staton:2017}.
\end{remark}

\subsubsection{Universal and existential types}

Our language contains both universal and existential types; it is common to
``encode'' the latter in terms of the former, but encodings of this form are
not quite correct as they neglect the equational theory of existential types.\footnote{The correct equational theory of polymorphically-encoded existentials does hold up to parametricity, but parametricity is an emergent property of \emph{syntax}. The purpose of specifying an equational theory is to constrain \emph{all} models, not only the nonstandard parametric models.}
We therefore include both connectives as primitives.
\begin{mathpar}
  \inferrule{
    \IsTp{\iota}{\Xi,\alpha}{A}
  }{
    \IsTp{\iota}{\Xi}{\forall \alpha. A}
  }
  \and
  \inferrule{
    \IsPTp{\Xi,\alpha}{A}
  }{
    \IsPTp{\Xi}{\exists \alpha. A}
  }
  \and
  \inferrule{
    \IsEl{\Xi,\alpha}{\Gamma}{u}{A}
  }{
    \IsEl{\Xi}{\Gamma}{\Lambda\alpha.u}{\forall\alpha.A}
  }
  \and
  \inferrule{
    \IsEl{\Xi}{\Gamma}{u}{\forall\alpha.A}
    \\
    \IsPTp{\Xi}{B}
  }{
    \IsEl{\Xi}{\Gamma}{u\cdot B}{A\brk{B/\alpha}}
  }
  \and
  \inferrule{
    \IsPTp{\Xi}{B}\\
    \IsEl{\Xi}{\Gamma}{u}{A\brk{B/\alpha}}
  }{
    \IsEl{\Xi}{\Gamma}{\TmPack\,\prn{B,u}}{\exists\alpha. A}
  }
  \and
  \inferrule{
    \IsPTp{\Xi}{C}\\
    \IsEl{\Xi}{\Gamma}{u}{\exists\alpha.A}
    \\
    \IsEl{\Xi,\alpha}{\Gamma,x:A}{v}{C}
  }{
    \IsEl{\Xi}{\Gamma}{
      \mathtt{let}\ \mathtt{pack}\,\prn{\alpha,x} = u\ \mathtt{in}\ v
    }{
      C
    }
  }
\end{mathpar}

The equational theory of universals and existentials is presented in \cref{sec:omitted-rules}.

%

\subsection{Rules for the propositional fragment}

The logical layer of \LMuRef{} is a form of \emph{guarded higher-order
separation logic}. We will treat each of these aspects modularly; in
\cref{sec:equational-logic,sec:ihol} we recall the rules of intuitionistic
higher-order equational logic, and in \cref{sec:separation-logic} we recall
(affine) separation logic, and we finish in \cref{sec:guarded-logic} with an
overview of the \emph{later modality} and its L\"ob induction principle, and
how they interact with the rest of the logic.

\subsubsection{Equational logic}\label{sec:equational-logic}

So far the only form of equality that we have considered is \emph{judgmental}
or \emph{external equality} in the sense of Jacobs~\cite[\S3.2]{jacobs:1999};
in order to facilitate equational reasoning in the logic, we add propositional
equality and relate it to judgmental equality using the following rules:
\begin{mathpar}
  \inferrule[equality formation]{
    \strut
    \IsTp{\iota}{\Xi}{A}\\
    \IsEl{\Xi}{\Gamma}{u,v}{A}
  }{
    \IsProp{\Xi}{\Gamma}{\EqProp{A}{u}{v}}
  }
  \and
  \begingroup
  \mprset{fraction={===}}
  \inferrule[Lawvere rule]{
    \strut
    \Entails{\Xi}{\Gamma,x:A}{\phi\brk{x/y}}{\psi\brk{x/y}}
  }{
    \Entails{\Xi}{\Gamma,x:A,y:A}{\phi\land\EqProp{A}{x}{y}}{\psi}
  }
  \endgroup
  \and
  \inferrule[equality reflection]{
    \Entails{\Xi}{\Gamma}{\top}{\EqProp{A}{u}{v}}
  }{
    \EqEl{\Xi}{\Gamma}{u}{v}{A}
  }
\end{mathpar}

As Jacobs~\cite{jacobs:1999} explains, the \textsc{Lawvere} rule entails all
important properties of equality, including congruence for all constructs of
\LMuRef{} and its \FMuRef{} substrate. The \textsc{equality reflection}
rule above is needed to complete the relationship between (unconditional)
judgmental equality and propositional equality without
assumptions.\footnote{Ordinarily, equality reflection would imply the
\textsc{Lawvere} rule but for the lack of propositional assumptions in
judgmental equality.}

\subsubsection{Intuitionistic higher-order logic}\label{sec:ihol}

We assume the usual rules of intuitionistic first-order logic over logical types; in
particular, in addition to implications $\Rightarrow$, conjunctions $\land$,
and disjunctions $\lor$, we may form universal and existential quantifications
$\forall\prn{x:A}.\phi$ and $\exists\prn{x:A}.\phi$ when $A$ is a logical type.
The logic is made higher-order by introducing a logical type classifying all
propositions.
\begin{mathpar}
  \inferrule{
    \strut
  }{
    \IsLTp{\Xi}{\TpProp}
  }
  \and
  \begingroup
  \mprset{fraction={===}}
  \inferrule{
    \IsEl{\Xi}{\Gamma}{\phi}{\TpProp}
  }{
    \IsProp{\Xi}{\Gamma}{\phi}
  }
  \and
  \inferrule{
    \EqEl{\Xi}{\Gamma}{\phi}{\psi}{\TpProp}
  }{
    \IsProp{\Xi}{\Gamma}{\phi\equiv \psi}
  }
  \endgroup
  \and
  \inferrule{
    \Entails{\Xi}{\Gamma}{\chi\land\phi}{\psi}\\
    \Entails{\Xi}{\Gamma}{\chi\land\psi}{\phi}
  }{
    \Entails{\Xi}{\Gamma}{\chi}{\EqProp{\TpProp}{\phi}{\psi}}
  }
\end{mathpar}

\subsubsection{Separation logic for local reasoning}\label{sec:separation-logic}

We assume the standard rules for intuitionistic affine separation logic, in
which we have a separating conjunction $\phi\SepConj\psi$ with unit $\top$, and
separating implications given as right adjoints $\prn{-\SepConj\psi} \dashv
\prn{\psi\SepImp-}$.
\begin{mathpar}
  \begingroup
  \mprset{fraction={===}}
  \inferrule{
    \BareEntails{\chi\SepConj\phi}{\psi}
  }{
    \BareEntails{\chi}{\phi\SepImp\psi}
  }
  \endgroup
  \and
  \inferrule{
    \BareEntails{\phi}{\phi'}
    \\
    \BareEntails{\psi}{\psi'}
  }{
    \BareEntails{\phi\SepConj\psi}{\phi'\SepConj\psi'}
  }
  \and
  \inferrule{}{
    \BareEntails{\phi\SepConj \psi}{\phi\land \psi}
    \\\\
    \BiBareEntails{\prn{\chi\SepConj\phi}\SepConj\psi}{\chi\SepConj\prn{\phi\SepConj\psi}}
    \\\\
    \BiBareEntails{\phi\SepConj\psi}{\psi\SepConj\phi}
  }
  \and
\end{mathpar}

In addition to the separating conjunction and implication, separation logic
contains a \emph{coreflective sublogic} of \DefEmph{persistent propositions},
which are to a first approximation those that are not sensitive to the state of
the heap and can therefore be duplicated freely. In particular, we add an
idempotent comonadic modality $\square$ that takes a proposition to its
``persistent core''; a proposition $\phi$ is then called persistent if the
entailment $\phi\vdash\square\phi$ holds. Persistent propositions in this sense
are closed under all the connectives of intuitionistic first-order logic; we
omit the rules that establish this and instead focus on the interaction between
persistence and the connectives of separation logic:
\begin{mathpar}
  \inferrule{
    \IsProp{\Xi}{\Gamma}{\phi}
  }{
    \IsProp{\Xi}{\Gamma}{\square\phi}
  }
  \and
  \inferrule{
    \BareEntails{\phi}{\psi}
  }{
    \BareEntails{\square\phi}{\square\psi}
  }
  \and
  \inferrule{}{
    \BareEntails{\square\phi}{\phi\SepConj\square\phi}
    \\\\
    \BareEntails{\square\phi}{\square\square\phi}
    \\\\
    \BareEntails{\phi\land\square\psi}{\phi\SepConj\square\psi}
  }
\end{mathpar}

It follows from the above that the separating conjunction of persistent
propositions is their conjunction.

\subsubsection{The later modality and guarded recursion}\label{sec:guarded-logic}

In order to use \LMuRef{} as a logic to reason about general recursion
(including recursion inherent in the heap), it is necessary to introduce the
\DefEmph{later modality} $\PropLtr$; as in prior
works~\cite{appel-mcallester:2001,dreyer-ahmed-birkedal:2009,iris:2018}, the
later modality abstracts away the onerous step-indices of more concrete
accounts of higher-order store leaving only the essential logical structure of
guarded-recursive reasoning. The abstract will meet the concrete, however, when
we illustrate below the interaction between the later modality and the constructions
of our programming language in the \textsc{fold equality} and \textsc{step equality}
rules.
\begin{mathpar}
  \inferrule{
    \IsProp{\Xi}{\Gamma}{\phi}
  }{
    \IsProp{\Xi}{\Gamma}{\PropLtr\phi}
  }
  \and
  \inferrule{}{
    \BareEntails{\phi}{\PropLtr{\phi}}
    \\\\
    \BareEntails{\PropLtr\phi\land\PropLtr\psi}{\PropLtr\prn{\phi\land\psi}}
  }
  \and
  \inferrule{}{
    \BareEntails{\PropLtr\phi\SepConj\PropLtr\psi}{\PropLtr\prn{\phi\SepConj\psi}}
    \\\\
    \BareEntails{\phi\SepImp\psi}{\PropLtr\phi\SepImp\PropLtr\psi}
    \\\\
    \BiBareEntails{\PropLtr\square\phi}{\square{\PropLtr\phi}}
  }
  \\
  \begingroup
  \mprset{fraction={===}}
  \inferrule[fold equality]{
    \BareEntails{\phi}{\PropLtr\prn{\EqProp{A\brk{\mu\alpha.A/\alpha}}{u}{v}}}
  }{
    \BareEntails{
      \phi
    }{
      \EqProp{\mu\alpha.A}{\TmFold\,u}{\TmFold\,v}
    }
  }
  \and
  \inferrule[step equality]{
    \BareEntails{\phi}{\PropLtr\prn{\EqProp{\TpT\,A}{u}{v}}}
  }{
    \BareEntails{\phi}{
      \EqProp{\TpT\,A}{\TmStep;u}{\TmStep;v}
    }
  }
  \endgroup
  \and
  \inferrule[L\"ob induction]{}{
    \BareEntails{\PropLtr\phi\Rightarrow\phi}{\phi}
  }
\end{mathpar}

The \textsc{L\"ob induction} rule above is what makes (guarded) recursive
reasoning possible in \LMuRef{}; the function of the \textsc{fold equality} and
\textsc{step equality} rules is to provide ``fuel'' that can be used to
discharge the later modality in the L\"ob induction hypothesis. This is the
sense in which \LMuRef{} evinces an abstract form of step-indexing: operations
that semantically involve unfolding a recursive domain equation leave behind
abstract steps that can be used to advance in time in relation to the later modality.

\subsubsection{Weakest preconditions}

For reasoning about programs, we introduce a connective called the
\DefEmph{partial weakest precondition}.  Morally, the weakest precondition of a
program and a predicate, as the name suggests, is the weakest proposition that
guarantees the predicate shall hold of any return value of the program. Note,
however, that, despite the name, we make no claim, neither in the logic nor in
its semantics, that it is in fact the weakest such proposition; this is in line
with the usage in Iris~\cite{iris:2018}.
\begin{mathpar}
  \inferrule[wp formation]{
    \IsPTp{\Xi}{A} \\
    \IsEl{\Xi}{\Gamma}{e}{\TpT\, A} \\
    \IsProp{\Xi}{\Gamma, x : A}{\phi}
  }{
    \IsProp{\Xi}{\Gamma}{\Lwp\, e\, \brc{x.\, \phi}}
  }
  \and
  \inferrule[wp-wand]{}{
    \BareEntails{\prn{\forall x.\, \phi \SepImp \psi} \SepConj \Lwp\, e\, \brc{x.\, \phi}}{\Lwp\, e\, \brc{x.\, \psi}}
  }
  \and
  \inferrule[wp-val]{}{
    \BareEntails{\phi \brk{e/x}}{\Lwp\, \prn{\TmRet\, e}\, \brc{x.\, \phi}}
  }
  \and
  \inferrule[wp-bind]{}{
    \BareEntails{\Lwp\, e\Sub{1}\, \brc{x.\, \Lwp\, e\Sub{2}\, \brc{y.\, \phi}}}{\Lwp\, \prn{x \gets e\Sub{1} ; e\Sub{2}}\, \brc{y.\, \phi}}
  }
  \and
  \inferrule[wp-get]{}{
    \BareEntails{\exists x.\, \PointsTo{\ell}{x} \SepConj \rhd \prn{\PointsTo{\ell}{x} \SepImp \phi}}{\Lwp\, \prn{\TmGet\, \ell}\, \brc{x.\, \phi}}
  }
  \and
  \inferrule[wp-set]{}{
    \BareEntails{\prn{\exists y.\, \PointsTo{\ell}{y}} \SepConj \prn{\PointsTo{\ell}{e} \SepImp \phi}}{\Lwp\, \prn{\TmSet\, \ell\, e}\, \brc{\_.\, \phi}}
  }
  \and
  \inferrule[wp-new]{}{
    \BareEntails{\forall x.\, \PointsTo{x}{e} \SepImp \phi}{\Lwp\, \prn{\TmNew\, e}\, \brc{x.\, \phi}}
  }
  \and
  \inferrule[wp-step]{}{
    \BareEntails{\rhd \phi}{\Lwp\, \TmStep\, \brc{\_.\, \phi}}
  }
\end{mathpar}

It is worth comparing the above rules to the rules in
Iris~\cite[Fig.13]{iris:2018}.  Our first four entailments are exactly the
same as the corresponding rules in Iris; note that as in Iris,
\textsc{wp-wand} implies the frame rule. Our rules for \textsc{wp-set} and
\textsc{wp-new} rules differ slightly from those of Iris, which have an
occurrence of the later modality in the antecedent that ours lack. This is
because every operation in Iris takes a step, but in our semantics, only
operations that semantically correspond to unfolding a recursive domain
equation do.

Values in the heap are stored ``one step in the future'', so $\TmGet$ must take
a step before returning; consequently, the postcondition of $\TmGet$ only needs
to hold \emph{later}. However, it is not an issue for $\TmSet$ and $\TmNew$ to
send a value to the future without going there, implying the postcondition must
be known now. This behaviour is reflected in the rules \textsc{get-after-set}
and \textsc{set-after-new}, where the former shows that $\TmGet$ takes a step,
whilst $\TmRet$ does not, and the latter shows that $\TmSet$ does not take a
step.
A second difference to point out is that whilst Iris has a rule for
$\beta$-reduction of functions, we do not. This is because our programming
language is subject to the $\beta/\eta$-equational theory of monadic
$\lambda$-calculus, so rather than having a rule stating that $\Lwp\, \prn{e
\brk{v/x}}\, \brc{y.\, \phi} \vdash \Lwp\, \prn{\prn{\lambda x.\, e}\, v}\,
\brc{y.\, \phi}$, the two propositions are actually \emph{convertible}.

\begin{example}
  It is not difficult to encode the more familiar \DefEmph{Hoare triples} in \LMuRef{},
  using the standard decomposition into persistence, separating implication, and weakest precondition:
  \[
    \brc{\phi}\, e\, \brc{x.\, \psi} \triangleq \square \prn{\phi \SepImp \Lwp\, e\,
    \brc{x.\, \psi}} .
  \]
\end{example}

\subsection{Recursive Functions}

As is standard, recursive types can be used to derive recursive terms via
self-referential types (see \eg Harper~\cite{harper:2016}). Using this approach, we
obtain terms $\TmRec f x e$ for recursive functions typed as follows:
\begin{mathpar}
  \inferrule {
    \IsPTp \Xi {A , B} \\
    \IsEl \Xi {\Gamma , f : A \to \TpT\, B , x : A} e {\TpT\, B}
  } {
    \IsEl \Xi \Gamma {\TmRec f x e} {A \to \TpT\, B}
  }
\end{mathpar}

This derived form satisfies the equation
$\prn{\TmRec f x e} \, x \equiv \mathtt{step} ; e \brk {\TmRec f x e / f}$.
L\"ob induction implies the following weakest precondition rule, reminiscent
of the corresponding Hoare triple rule in Iris~\cite{birkedal-bizjak:2022:iris}:
\begin{mathpar}
  \inferrule[wp-rec]{
    \IsPTp \Xi {A , B} \\
    \IsLTp \Xi C \\
    \IsEl \Xi {\Gamma , z : C , f : A \to \TpT\, B , x : A} e {\TpT\, B} \\
    \IsProp \Xi \Gamma \phi \\
    \IsProp \Xi {\Gamma , z : C , x : A} \psi \\
    \IsProp \Xi {\Gamma , z : C , x : A , y : B} \chi \\
    \Entails \Xi \Gamma {\phi \land \forall z .\, \forall x .\, \psi
      \SepImp \Lwp\, \prn{\prn{\TmRec f x e} \, x}\, \brc {y .\, \chi}} {\forall z .\,
      \forall x .\, \psi \SepImp \Lwp\, \prn{e \brk {\TmRec f x e / f}}\,
      \brc {y .\, \chi}}
  } {
    \Entails \Xi \Gamma {\rhd \phi} {\forall z .\, \forall x .\,
      \psi \SepImp \Lwp\, \prn{\prn{\TmRec f x e} \, x}\, \brc {y .\, \chi}}
  }
\end{mathpar}

\section{Case study: verifying the \emph{append} function on linked lists}\label{sec:example}

In this section, we illustrate the use of \LMuRef{} by an elementary case
study: linked lists in the heap and their append function. The proof is very
similar to the one in Iris~\cite[\S4.2]{birkedal-bizjak:2022:iris}, although
some steps are perhaps slightly simpler as they use equational reasoning rather
than explicit rules for reduction in weakest preconditions. We first define a
recursive type of imperative linked lists on any type $\alpha$:
\[
  \TpLList \alpha \triangleq \mu \rho . \mathtt 1 + \mathtt{ref}\, \prn{\alpha \times \rho} .
\]

Our goal is to define the \emph{append} function on imperative linked lists and
prove that it is correct.  This means, in particular, to show that it behaves
the same as a pure \emph{reference implementation} defined on functional lists.
In order to say what it means for a function on imperative lists to behave like
a function on linked lists, we must first introduce a formal correspondence
between the two types~\cite[Sec.4.2]{birkedal-bizjak:2022:iris}. This can be
done directly as a structurally recursive function in \LMuRef{}.

\NewDocumentCommand\ListInv{}{\mathrel{\approx}}

\begin{construction}[The list invariant]
  We define a correspondence $\prn{\ListInv} : \TpLList{\alpha}\to\TpList{\alpha}\to\TpProp$ in
  \LMuRef{} by structural recursion on the second argument.

  \iblock{
    \mrow{
      \prn{\ListInv} : \TpLList{\alpha} \to \TpList{\alpha} \to \mathtt{prop}
    }
    \mrow{
      l \ListInv \TmNil \triangleq \prn{l = \mathtt{fold} \, \prn{\mathtt{inl}\, \prn{}}}
    }
    \mhang{
      l \ListInv \prn{\TmCons{x}{\mathit{xs}}} \triangleq
    }{
      \mrow{\exists\prn{r : \mathtt{ref}\, \prn{A \times \TpLList{\alpha}}}.}
      \mrow{\exists\prn{l' : \TpLList{\alpha}}.}
      \mrow{
        \prn{l = \mathtt{fold} \, \prn{\mathtt{inr} \, r}}
        \SepConj \PointsTo{r}{\prn{x , l'}}
        \SepConj \prn{l'\ListInv \mathit{xs}}
      }
    }
  }
\end{construction}

Note that whilst elements of type $\TpLList{\alpha}$ could potentially be cyclic,
this is ruled out by the list invariant above: the separating conjunction
consumes the location $r$ so it cannot appear again, and furthermore, a cyclic
list would be infinite and therefore cannot correspond to a functional list.

\begin{construction}[The \emph{append} function]
  We define the append function on linked lists and its pure reference
  implementation on functional lists below.

  \iblock{
    \mrow{
      \LAppend : \TpLList{\alpha} \times \TpLList{\alpha} \to \TpT\, \prn{\TpLList{\alpha}}
    }
    \mhang{
      \LAppend \triangleq
    }{
      \mhang{\TmRec{f}{l_1, l_2}{}}{
        \mrow{z \gets \mathtt{unfold}\, l_1;}
        \mhang{\mathtt{match}\, z\, \mathtt{with}}{
          \mrow{\mathtt{inl}\, \_ \Rightarrow \mathtt{ret} \, l_2}
          \mrow{\mathtt{inr} \, r \Rightarrow \prn{a , l_1'} \gets \mathtt{get} \,
          r ; l_3 \gets f\, \prn{l_1', l_2} ; \mathtt{set} \, r \, \prn{a , l_3}
          ; \mathtt{ret}\, \prn{\mathtt{fold}\,\prn{\mathtt{inr} \, r}}}
        }
      }
    }
  }

  The function above is both impure and general recursive. By contrast, the
  reference implementation below is pure and structurally recursive.

  \iblock{
    \mrow{\prn{\oplus} : \TpList{\alpha} \to \TpList{\alpha} \to \TpList{\alpha}}
    \mrow{\TmNil \oplus \mathit{ys} \triangleq \mathit{ys}}
    \mrow{\prn{\TmCons{x}{\mathit{xs}}}\oplus \mathit{ys} \triangleq \TmCons{x}{\prn{\mathit{xs} \oplus \mathit{ys}}}}
  }
\end{construction}

We can now prove that $\LAppend$ behaves according to the reference
implementation $\prn{\oplus}$.

\begin{theorem}
  The following sequent is derivable in $\LMuRef$:
  \[
    \alpha\mid\cdot\mid\top\vdash
    \forall\prn{u_1,u_2:\TpList{\alpha}; l_1,l_2:\TpLList{\alpha}}.\,
    l_1\ListInv u_1 \SepConj l_2\ListInv u_2
    \SepImp
    \Lwp\, \prn{\LAppend\, \prn{l_1 , l_2}}\, \brc {x .\, x\ListInv {u_1 \oplus u_2}}
  \]
\end{theorem}

\begin{proof}
  Let $\brk{\LAppend}$ be the function defined so
  that $\LAppend\,\prn{l_1,l_2}\equiv \TmRec{f}{l_1,l_2}{\brk{\LAppend}\,f\,\prn{l_1,l_2}}$,
  and let $Q$ be the following predicate:
  \begin{align*}
    Q &: \prn{\TpLList{\alpha}\times\TpLList{\alpha}\to \TpT\,\prn{\TpLList{\alpha}}} \to \TpProp
    \\
    Q\,f &\triangleq
    \forall\prn{u_1,u_2:\TpList{\alpha}; l_1,l_2:\TpLList{\alpha}}.\,
    l_1\ListInv u_1 \SepConj \l_2\ListInv u_2
    \SepImp
    \Lwp\, \prn{f\,\prn{l_1,l_2}}\, \brc {x .\, x\ListInv {u_1 \oplus u_2}}
  \end{align*}

  Our goal is to prove $Q\,\LAppend$; applying the \textsc{wp-rec} rule, it
  suffices to show that $Q\,\LAppend\vdash Q\,\prn{\brk{\LAppend}\,\LAppend}$.
  We proceed by cases on $u_1$; the only non-trivial
  case is the following:
  $
    \BareEntails{
      Q\,\LAppend
      \SepConj
      l_1\ListInv \TmCons{v}{\mathit{vs}} \SepConj l_2\ListInv u_2
    }{
      \Lwp\, \prn{\brk{\LAppend}\,\LAppend\,\prn{l_1,l_2}}\, \brc {x.\, x\ListInv \TmCons{v}{\prn{\mathit{vs} \oplus u_2}}}
    }
  $.
  Rewriting by the defining clause of $\prn{\ListInv}$, we may
  assume $r:\TpRef\,\prn{A\times \TpLList{\alpha}}$ and $s:\TpLList{\alpha}$ to prove the following:

  \iblock{
    \mhang{
      Q\,\LAppend
      \SepConj
      \prn{l_2\ListInv u_2}
      \SepConj
      \PointsTo{r}{\prn{v,s}}
      \SepConj
      \prn{s\ListInv\mathit{vs}}
      \vdash
    }{
      \mrow{
        \Lwp\, \prn{\brk{\LAppend}\,\LAppend\,\prn{\TmFold\,\prn{\mathtt{inr}\,r},l_2}}\, \brc {x.\, x\ListInv \TmCons{v}{\prn{\mathit{vs} \oplus u_2}}}
      }
    }
  }

  Applying equational reasoning (including \textsc{unfold-of-fold}),
  we convert our goal to the following:

  \iblock{
    \mhang{
      Q\,\LAppend
      \SepConj
      \prn{l_2\ListInv u_2}
      \SepConj
      \PointsTo{r}{\prn{v,s}}
      \SepConj
      \prn{s\ListInv\mathit{vs}}
      \vdash
    }{
      \mrow{
        \Lwp\, \prn{
          \TmStep;
          \prn{a,l_1'}\gets\TmGet\,{r};
          l_3\gets \LAppend\,\prn{l'_1,l_2};\TmSet\,r\,\prn{a,l_3};\TmRet\,\prn{\TmFold\,\prn{\mathtt{inr}\,r}}
        }\, \brc {x.\, x\ListInv \TmCons{v}{\prn{\mathit{vs} \oplus u_2}}}
      }
    }
  }

  Applying \textsc{wp-step} and the introduction rule for the later
  modality, it suffices to prove:

  \iblock{
    \mhang{
      Q\,\LAppend
      \SepConj
      \prn{l_2\ListInv u_2}
      \SepConj
      \PointsTo{r}{\prn{v,s}}
      \SepConj
      \prn{s\ListInv\mathit{vs}}
      \vdash
    }{
      \mrow{
        \Lwp\, \prn{
          \prn{a,l_1'}\gets\TmGet\,{r};
          l_3\gets \LAppend\,\prn{l'_1,l_2};\TmSet\,r\,\prn{a,l_3};\TmRet\,\prn{\TmFold\,\prn{\mathtt{inr}\,r}}
        }\, \brc {x.\, x\ListInv \TmCons{v}{\prn{\mathit{vs} \oplus u_2}}}
      }
    }
  }

  Repeatedly applying weakest precondition rules and other administrative rules, it suffices to prove:

  \iblock{
    \mhang{
      Q\,\LAppend
      \SepConj
      \prn{l_2\ListInv u_2}
      \SepConj
      \prn{s\ListInv\mathit{vs}}
      \SepConj
      \PointsTo{r}{\prn{v,s}}
      \vdash
    }{
      \mrow{
        \Lwp\,\prn{\LAppend\,\prn{s,l_2}}\,\brc{
          l_3.\,
          \Lwp\,\prn{\TmSet\,r\,\prn{v,l_3}}\,\brc{
            \_.\,
            {\TmFold\,\prn{\mathtt{inr}\,r}}\ListInv \TmCons{v}{\prn{\mathit{vs} \oplus u_2}}
          }
        }
      }
    }
  }

  Next we use $Q\,\LAppend$ and \textsc{wp-wand} to reduce our goal as
  follows, fixing $l_3:\TpLList{\alpha}$:

  \iblock{
    \mrow{
      \PointsTo{r}{\prn{v,s}}
      \SepConj
      \prn{l_3\ListInv \mathit{vs} \oplus u_2}
      \vdash
      \Lwp\,\prn{\TmSet\,r\,\prn{v,l_3}}\,\brc{
        \_.\,
        {\TmFold\,\prn{\mathtt{inr}\,r}}\ListInv \TmCons{v}{\prn{\mathit{vs} \oplus u_2}}
      }
    }
  }

  Applying \textsc{wp-set}, we arrive at the following goal:

  \iblock{
    \mrow{
      \PointsTo{r}{\prn{v,l_3}}
      \SepConj
      \prn{l_3\ListInv \mathit{vs} \oplus u_2}
      \vdash
      {\TmFold\,\prn{\mathtt{inr}\,r}}\ListInv \TmCons{v}{\prn{\mathit{vs} \oplus u_2}}
    }
  }

  The above is immediate by definition of $\prn{\ListInv}$ and instantiation of existential variables.
\end{proof}

\section{Denotational semantics of \LMuRef{} in impredicative guarded dependent type theory}

The denotational semantics of \LMuRef{} is an extension of the model of
System~\FMuRef{} previously constructed by Sterling, Gratzer, and
Birkedal~\cite{sterling-gratzer-birkedal:2022}. For lack of space, we can only
give a brief introduction to the latter, focusing on the main ideas. The main
ingredient to our semantics is the use of \DefEmph{impredicative guarded
  dependent type theory} (\iGDTT{}) as a sufficiently powerful metalanguage to
admit both the synthetic solution of domain equations for recursively defined
semantic worlds (which uses guarded recursion) \emph{and} the definition of the
store-passing monad (which uses impredicativity). In \cref{sec:igdtt} we give a
brief overview of \iGDTT{}, and we proceed in \cref{sec:semantics-of-store} to
explain the interpretation of higher-order store.

\subsection{Impredicative guarded dependent type theory}\label{sec:igdtt}

Impredicative guarded dependent type theory or \iGDTT{} is roughly the
extension of extensional \DefEmph{guarded dependent type
theory}~\cite{bgcmb:2016,bizjak-mogelberg:2020} by additional (impredicative)
universe structure. We first describe the universes in
\cref{sec:impredicativity}, and then briefly explain the basics of \iGDTT{}'s
synthetic guarded domain theory in
\cref{sec:guarded,sec:guarded-domain-theory}. Finally, we describe a simple
recipe for constructing models of \iGDTT{} in \cref{sec:igdtt-consistency} from
which consistency immediately follows.

\subsubsection{Impredicative universes}\label{sec:impredicativity}

\iGDTT{} adds to ordinary guarded dependent type theory the following
additional universe structure:

\begin{enumerate}

\item We assume an ordinary hierarchy of predicative universes $\UniType_0:
  \UniType_1: \ldots$; when it causes no confusion, we will write $\UniType$ for any
  appropriate $\UniType_i$.

\item We further assume a pair of impredicative universes $\UniSet,\UniP:\UniType_0$
  of small types and proof-irrelevant propositions respectively, where the
  latter satisfies propositional extensionality.\footnote{Note that $\UniP$
    is not an element of $\UniSet$, as this would be
    inconsistent~\cite{coquand:1986}.} Finally, we assume that any element of
  $\UniP$ is also an element of $\UniSet$.

\end{enumerate}

The universe structure above is roughly that of the $\mathtt{Set}$ and
$\mathtt{Prop}$ universes of Coq~\cite{coq:reference-manual} underneath
$\mathtt{Type}$ when the \verb|-impredicative-set| option is activated.
Impredicativity of $\UniSet$ means closure under dependent products
$\Hl{\Forall{x:A}{B}:\UniSet}$ of ``large-indexed'' families $x : A \vdash
B:\UniSet$ when $A:\UniType$, and likewise for $\UniP$. The coherent
impredicative encoding of Awodey, Frey, and
Speight~\cite{awodey-frey-speight:2018} ensures that the full internal
subcategory determined by $\UniSet$ is \emph{reflective} in $\UniType$, and so
a genuine existential $\Hl{\Exists{x:A}B:\UniSet}$ can be obtained by applying
the reflection to the (large) dependent sum $\Hl{\Sum{x:A}{B}:\UniType}$. On
$\UniP$ these are exactly the ordinary universal and existential quantifiers
--- as the reflection is the ``bracket
type'' of Awodey and Bauer~\cite{awodey-bauer:2004}.

In what follows, we shall let $\UU$ stand for any of the universes of
\iGDTT{} described here.

\subsubsection{The later modality and guarded recursion}\label{sec:guarded}

Every universe $\UU$ is closed under a
``later modality'' $\Hl{\Ltr : \UU\to \UU}$ facilitating guarded recursion.
The later modality also satisfies a dependently typed version of the rules of
an applicative functor. Although there are many ways to present this structure,
we choose to follow prior
work~\cite{bgcmb:2016,bizjak-mogelberg:2020,sterling-gratzer-birkedal:2022} by
formulating them using \DefEmph{delayed substitutions}
\JdgBox{\delta\DSubst\Delta}, which we describe simultaneously with the rules
of the later modality:
\begin{mathpar}
  \inferrule[later formation]{
    \delta \DSubst \Delta\\
    \Delta \vdash A\ \mathit{type}
  }{
    \Ltr\brk{\delta}.A\ \mathit{type}
  }
  \and
  \inferrule[later functoriality]{
    \delta\DSubst\Delta\\
    \Delta\vdash a:A
  }{
    \Next*\brk{\delta}.{a} : \Ltr\brk{\delta}.{A}
  }
  \and
  \inferrule[empty dsubst.]{
    \vphantom{\Delta}
  }{
    \cdot\DSubst\cdot
  }
  \and
  \inferrule[extended dsubst.]{
    \delta \DSubst \Delta\\
    a : \Ltr\brk{\delta}.A
  }{
    \prn{\delta, x\leftarrow a} \DSubst \Delta, x:A
  } \and
  \inferrule[later weakening]{
    \delta \DSubst \Delta \\
    a : \Ltr \brk{\delta} . A \\
    \Delta \vdash B\ \mathit{type}
  }{
    \Ltr \brk{\delta, x \gets a} . B = \Ltr \brk{\delta} . B
  } \and
  \inferrule[next weakening]{
    \delta \DSubst \Delta \\
    a : \Ltr \brk{\delta} . A \\
    \Delta \vdash b : B
  }{
    \Next* \brk{\delta , x \gets a} . b = \Next* \brk{\delta} . b
  } \and
  \inferrule[later exchange]{
    \delta \DSubst \Delta \\
    a : \Ltr \brk{\delta} . A \\
    b : \Ltr \brk{\delta} . B \\
    \Delta , x : A , y : B \vdash \delta' \DSubst \Delta' \\
    \Delta , x : A , y : B , \Delta' \vdash C\ \mathit{type}
  }{
    \Ltr \brk{\delta , x \gets a , y \gets b , \delta'} . C = \Ltr \brk{\delta , y \gets b , x \gets a , \delta'} . C
  } \and
  \inferrule[next exchange]{
    \delta \DSubst \Delta \\
    a : \Ltr \brk{\delta} . A \\
    b : \Ltr \brk{\delta} . B \\
    \Delta , x : A , y : B \vdash \delta' \DSubst \Delta' \\
    \Delta , x : A , y : B , \Delta' \vdash c : C
  }{
    \Next* \brk{\delta , x \gets a , y \gets b , \delta'}  . c = \Next* \brk{\delta , y \gets b , x \gets a , \delta'} . c
  } \and
  \inferrule[later force]{
    \delta \DSubst \Delta \\
    \Delta \vdash a : A \\
    \Delta , x : A \vdash B\ \mathit{type}
  }{
    \Ltr \brk{\delta , x \gets \Next* \brk{\delta} . a} . B = \Ltr \brk{\delta} . B \brk{a / x}
  } \and
  \inferrule[next force]{
    \delta \DSubst \Delta \\
    \Delta \vdash a : A \\
    \Delta , x : A \vdash b : B
  }{
    \Next* \brk{\delta , x \gets \Next* \brk{\delta} . a} . b = \Next* \brk{\delta} . b \brk{a / x}
  } \and
  \inferrule[later id]{
    \delta \DSubst \Delta \\
    \Delta \vdash a : A \\
    \Delta \vdash a' : A
  }{
    \prn {\Next* \brk{\delta} . a = \Next* \brk{\delta} . a'} = \Ltr \brk{\delta} . \prn {a = a'}
  } \and
  \inferrule[next variable]{
    \delta \DSubst \Delta \\
    a : \Ltr \brk{\delta} . A
  }{
    \Next* \brk{\delta , x \gets a} . x = a
  }
\end{mathpar}

Then the ordinary later modality $\Ltr : \UU\to \UU$ sends $A:\UU$ to $\Ltr\brk{\cdot}.A$
via the empty delayed substitution; likewise, we shall write $\Next{a}$ for
$\Next*\brk{\cdot}.a$. From these rules, we may deduce that the later modality
forms a \DefEmph{well-pointed endofunctor} $\Ltr:\UU\to\UU$ in the sense of
Kelly~\cite{kelly:1980}, and moreover preserves cartesian products.
The later modality also comes
equipped with a \DefEmph{L\"ob recursor} for defining guarded fixed points as
specified below:
\begin{mathpar}
  \inferrule[L\"ob recursor]{}{
    \Con{gfix} : \prn{\Ltr{A}\to A}\to A
  }
  \and
  \inferrule[L\"ob unfolding]{}{
    \Con{gfix}\,f = f\prn{\Next\prn{\Con{gfix}\,f}}
  }
\end{mathpar}

\subsubsection{Basic synthetic guarded domain theory}\label{sec:guarded-domain-theory}

\begin{definition}\label[definition]{def:guarded-domain}
  A \DefEmph{guarded domain} in $\UU$ is defined to be an algebra for the
  endofunctor $\Ltr:\UU\to\UU$, \ie a type $X:\UU$ equipped with a function
  $\vartheta_X : \Ltr{A}\to A$. A homomorphism from $X$ to $Y$ is then given by
  a function $f : X\to Y$ that commutes with the algebra maps in the sense that
  $\vartheta_Y\circ \Ltr{f} = f\circ\vartheta_X$. We will write $\LtrAlg{\UU}$ for the
  category of guarded domains in $\UU$ and their homomorphisms.
\end{definition}

\begin{example}\label[example]{ex:universe-guarded-domain}
  The universe $\UU$ is a guarded domain in any higher universe $\VV$ as we may
  define $\vartheta_{\UU}:\Ltr{\UU}\to\UU$ to send $A:\Ltr{\UU}$ to the
  delayed type $\Ltr\brk{Z\leftarrow A}.Z$ using the unary delayed substitution
  $\Hl{\prn{Z\leftarrow A}\DSubst Z:\UU}$.
\end{example}

Following Birkedal and M\o{}gelberg~\cite{birkedal-mogelberg:2013}, the L\"ob
recursor can be used to solve domain equations by computing fixed
points on the universe $\UU$. The simplest example of a guarded domain equation
is the one that defines the \DefEmph{guarded lift functor} $\MonadL : \UU\to
\LtrAlg{\UU}$ of Paviotti, M\o{}gelberg, and
Birkedal~\cite{paviotti-mogelberg-birkedal:2015} which sends a type to the free
guarded domain on that type, \ie the left adjoint to the forgetful functor from
guarded domains to types. The domain equation in question is $\MonadL{A} = A+
\Ltr{\MonadL{A}}$, which we solve using the L\"ob recursor on the universe
together with the latter's guarded domain structure:
\begin{align*}
  \MonadL{A} &\triangleq \Con{gfix}\,\prn{\lambda X : \Ltr{\UU}. A + \vartheta_{\UU} X}
  = A + \vartheta_{\UU}\prn{\Next\prn{\MonadL{A}}}
  = A + \Ltr\brk{Z\leftarrow \Next\prn{\MonadL{A}}}.Z
  = A + \Ltr{\MonadL{A}}
\end{align*}

The algebra structure $\vartheta\Sub{\MonadL{A}} :
\Ltr{\MonadL{A}}\to\MonadL{A}$ is given by the right-hand coproduct injection;
the left-hand injection then defines the unit map $\eta : A\to\MonadL{A}$ for
the monad determined by the resulting adjunction between guarded domains and
types.

\begin{construction}[Guarded domains are lift-algebras]
  Any guarded domain $X$ is also an algebra for the monad $\MonadL$. As
  $\MonadL{X}$ is the \emph{free} $\Ltr$-algebra on $X$, there is a unique
  homomorphism of $\Ltr$-algebras $\alpha_X:\MonadL{X}\to X$ such that
  $\alpha_X\circ \eta = \Con{id}_{X}$ which induces an
  $\MonadL$-algebra structure on $X$.
\end{construction}

\begin{construction}[Family lifting for the guarded lift monad]
  \label[construction]{cnstr:guarded-family-lifting}
  Let $\UU$ be a universe, and let $A:\UniType$ be a type. We may lift a
  family $\Phi : A\to \UU$ to a family $\Phi\Sup{\MonadL} : \MonadL{A}\to \UU$
  defined using the induced $\MonadL$-algebra structure on $\UU$, \ie
  $\Phi\Sup{\MonadL} \triangleq \alpha\Sub{\UU} \circ \MonadL{\Phi}$.
  We will occasionally write $\Hl{u \Downarrow \Phi}$ to mean $\Phi\Sup{\MonadL} u$.
\end{construction}

\subsubsection{Consistency and models of \iGDTT{}}\label{sec:igdtt-consistency}

A simple and modular recipe for constructing non-trivial models of \iGDTT{} is
provided by Sterling, Gratzer, and
Birkedal~\cite{sterling-gratzer-birkedal:2022}, from which consistency is
easily deduced.

\begin{theorem}[S., Gratzer, and B.~\cite{sterling-gratzer-birkedal:2022}]
  \label[theorem]{thm:igdtt-model}
  Let $\prn{\mathbb{O},\leq,\prec}$ be a \emph{separated intuitionistic
    well-founded poset}\footnote{We omit the definition of \emph{separated intuitionistic well-founded posets} for brevity and refer the
    reader to Sterling, Gratzer, and
    Birkedal~\cite{sterling-gratzer-birkedal:2022} for details.} in a realizability topos $\mathscr{S}$. Then internal presheaves $\brk{\OpCat{\mathbb{O}},\mathscr{S}}$ give a model of
  \iGDTT{} in which:
  \begin{enumerate}

  \item the predicative universes $\UniType$ are modeled by the Hofmann--Streicher
    liftings~\cite{hofmann-streicher:1997,awodey:2022:universes} of the
    universes of (small) assemblies~\cite{luo:1994} from $\mathscr{S}$;

  \item the impredicative universes $\UniP,\UniSet:\UniType$ are modeled by the
    Hofmann--Streicher liftings of the universes of $\lnot\lnot$-closed
    propositions and of modest sets in $\mathscr{S}$ respectively;

  \item the later modality $\Ltr$ is computed explicitly by the limit
    $\prn{\Ltr{A}}u = \Lim{v\prec u}Av$.

  \end{enumerate}
\end{theorem}

\begin{example}
  The simplest example of a model of \iGDTT{} instantiating
  \cref{thm:igdtt-model} is given by the standard order of the natural numbers
  object in Hyland's~\cite{hyland:1982} effective topos $\mathbf{Eff}$; this is
  exactly the ``topos of trees''~\cite{bmss:2011} constructed internally to
  $\mathbf{Eff}$. This model can be adjusted in two orthogonal directions, by
  varying the underlying partial combinatory algebra and by varying the
  internal well-founded order.
\end{example}

\begin{corollary}[S., Gratzer, and B.~\cite{sterling-gratzer-birkedal:2022}]
  \iGDTT{} is consistent.
\end{corollary}

\subsection{Denotational semantics of general store in \iGDTT{}}\label{sec:semantics-of-store}

We briefly recall the denotational semantics of general store in \iGDTT{} via a
\emph{presheaf model}. We will first construct a preorder $\Worlds$ of semantic
worlds (representing heap layouts), and then we shall interpret \emph{program
types} as $\UniSet$-valued co-presheaves on $\Worlds$ and \emph{logical types}
as $\UniType$-valued co-presheaves on $\Worlds$.

\begin{notation}
  When $P$ is a partial order and $p\leq_P q$, we will write $q_* : Ep \to Eq$
  for the covariant functorial action of any functor $E : P \to \mathscr{E}$.
\end{notation}

\subsubsection{Recursively defined semantic worlds}

In this section, we will define a partial order $\prn{\Worlds, \leq}$ of
semantic worlds simultaneously with the collection of semantic types by
solving a guarded domain equation in $\UniType$. We will ultimately define
$\Worlds$ to be a kind of finite mapping of locations to types, but we must
be more careful than usual because notions like ``finite subtype'' are
somewhat sensitive when $\UniP$ is not a true subobject classifier.

\begin{definition}
  Let $I$ be a totally ordered type; a \DefEmph{finite subtype} $U\FinSubset
  I$ is an element $\vrt{U} : \mathbb{N}$ together with a monotone
  injective\footnote{We mean injective in the general intuitionistic sense:
    elements of the domain are equal if and only if they are identified by the
    function in question.} function $\sigma_U:\Fin{\vrt{U}}\hookrightarrow I$. We
  will often abuse notation by writing $U$ to refer to the image of
  $\sigma_U$ in $I$. There is a (decidable) preorder on finite subtypes given
  by inclusion, which is in fact a partial order because of the monotonicity of
  $\sigma_U$ in the total order $I$.
\end{definition}

\begin{definition}
  Given a totally ordered type $I$, a \DefEmph{finite mapping} $w:I\finto T$
  is given by a finite subtype $\vrt{w}\FinSubset I$ called the
  \DefEmph{support} together with a function $\tau_w : \vrt{w}\to T$ called
  the \DefEmph{labeling}. Given $i\in\vrt{w}$ we shall simply write $wi:T$
  for $\tau_w i$.
  There is a partial order on finite mappings $w:I\finto T$ given by
  inclusion of supports: we say that $w\leq w'$ when $\vrt{w} \leq \vrt{w'}$
  and the restriction of $\tau\Sub{w'}$ to $\vrt{w}$ is equal to $\tau_w$.
  Note that the partial order on finite mappings is not decidable unless $T$
  has decidable equality.
\end{definition}

We can now use the notion of finite mapping to define a partial order of
\DefEmph{semantic worlds} $\Worlds$ simultaneously with the categories
$\SemPType,\SemLType$ of semantic program types and semantic logical types
respectively by solving the following guarded domain equation:
\[
  \Worlds = \mathbb{N}\finto \Ltr{\SemPType}
  \qquad
  \SemPType = \brk{\Worlds, \UniSet}
  \qquad
  \SemLType = \brk{\Worlds, \UniType}
\]
In the above, we have defined a world to be a finite mapping from memory
locations to delayed semantic program types, which are defined to be
\emph{$\UniSet$-valued co-presheaves} on the poset of semantic worlds.
Semantic logical types are defined similarly as $\UniType$-valued
co-presheaves.


The impredicativity of $\UniSet$ is essential for $\SemPType$ to be cartesian
closed and therefore capable of modelling function types. This is because exponentials
of co-presheaves are defined from natural transformations. In $\SemPType$, these
are $\Worlds$-indexed, and thus, since $\Worlds$ is not $\UniSet$-small, were it
not for the impredicativity of $\UniSet$, the homsets would also not be $\UniSet$-small.
The cartesian closure of $\SemLType$ follows simply from $\Worlds$ being
$\UniType$-small.


\begin{observation}
  Note that both $\SemPType$ and $\SemLType$ are guarded domains in the sense
  of \cref{def:guarded-domain}: the structure map
  $\vartheta\Sub{\SemType{\iota}}$ sends $A:\Ltr\SemType{\iota}$ to
  the co-presheaf $w\mapsto \Ltr\brk{Z\leftarrow A}. Zw$.
\end{observation}

\subsubsection{Semantic heaplets; total heaps and partial heaps}

The semantic notion of \emph{heap} or \emph{memory} can be specified in terms
of the more general \DefEmph{heaplet distributor} on $\Worlds$. This is the
distributor $\Hl{\Heaplet :
  \OpCat{\Worlds}\times\Worlds\to \UniSet}$ that classifies heaps whose layout is
governed by the contravariant parameter and whose values vary in the covariant
parameter.
\[
  \Heaplet\prn{w^-,w^+} \triangleq
  \Prod{l\in \vrt{w^-}}
  \vartheta\Sub{\SemPType}\prn{w^-l}\, w^+
  =
  \Prod{l\in\vrt{w^-}}
  \Ltr\brk{Z\leftarrow w^-l}. Zw^+
\]

Using the notion of a heaplet, we can define \DefEmph{partial heaps}
and \DefEmph{total heaps} and at a given world; the latter are used to
interpret the state monad of \FMuRef{}, whereas the former are used to
interpret the program logic. Partial heaps will be functorial in worlds,
whereas total heaps are a non-functorial derived form.

\begin{definition}
  A \DefEmph{partial heap} $h$ at a world $w$ is given by a finite subtype
  $\vvrt{h}\FinSubset\vrt{w}$ together with a heaplet
  $\eta_h:\Heaplet\prn{w\Sub{\vvrt{h}},w}$. We shall write $\vrt{h}\triangleq
  w\Sub{\vvrt{h}}$ for the supporting world; given $l\in \vvrt{h}$ we shall
  write $hl$ for $\eta_hl$.  Partial heaps are arranged into a functor
  $\Hl{\ParHeap : \Worlds\to\UniSet}$; the covariant functoriality in worlds
  $w\leq w'$ takes a partial heap $h:\ParHeap w$ to $w'_*h \triangleq
  \prn{\vvrt{h}, w'_*\eta_h}$.
\end{definition}

\begin{definition}
  Two partial heaps are \DefEmph{disjoint} from each other when their
  supports do not intersect. This property is both decidable and functorial
  in $\Worlds$, so it yields a decidable subobject of
  $\ParHeap\times\ParHeap$ in $\SemLType$. In the internal language of
  $\SemLType$, we will write $h \# h'$ to mean that $h$ and $h'$ are
  disjoint.
\end{definition}

\begin{construction}
  We may define an internal \DefEmph{partial commutative monoid} structure
  on $\ParHeap$ in $\SemLType$. The unit is the empty heap $\emptyset$,
  and the partial multiplication $h_1\cdot h_2$ is defined when $h_1 \# h_2$
  as follows:
  \begin{align*}
    \vvrt{h_1 \cdot_w h_2} &\triangleq \vvrt{h_1} \cup \vvrt{h_2}
    \\
    \eta\Sub{h_1\cdot_w h_2}\ell &\triangleq
    \begin{cases}
      \vrt{h_1\cdot_w h_2}_*\eta_{h_1}\ell &\text{ if $\ell\in \vvrt{h_1}$}
      \\
      \vrt{h_1\cdot_w h_2}_*\eta_{h_2}\ell &\text{ if $\ell\in \vvrt{h_2}$}
    \end{cases}
  \end{align*}
\end{construction}

\begin{definition}
  A \DefEmph{total heap} at a world $w$ is a partial heap $h:\ParHeap w$ such that
  $\vrt{h}=w$; this is not functorial in $\Worlds$, but we may write total
  heaps as a functor  $\Hl{\TotalHeap_{\bullet} : \vrt{\Worlds}\to \UniSet}$
  where $\vrt{\Worlds}$ is the underlying discrete category of $\Worlds$. We
  shall abusively write $\AllHeap$ for the dependent sum
  $\Sum{w:\Worlds}\TotalHeap_{w}$ that bundles a heap with its world.  We will
  write $\TotalHeap\Sub{\# w}$ to mean the type of total heaps whose support is
  disjoint from $w$.
\end{definition}
\subsubsection{Semantic domains for predicates}\label{sec:predicate-domains}

Here we describe the semantic domains that govern predicates and entailments;
these domains will ultimately form the basis for a BI-hyperdoctrine $\BIHyp$
over $\SemLType$, to be described later. We shall denote by $\WProp$ the
Hofmann--Streicher lifting of $\UniP$ into $\SemLType =
\brk{\Worlds,\UniType}$, defining $\BIAlg$ to be the internal poset of
$\WProp$-valued co-presheaves on the partial commutative monoid $\ParHeap$
under its extension order:
\begin{gather*}
  \begin{array}{l}
    \WProp : \SemLType\\
    \WProp w \triangleq \brk{w\downarrow \Worlds, \UniP}
  \end{array}
  \qquad
  \begin{array}{l}
    \BIAlg : \SemLType\\
    \BIAlg \triangleq \brk{\ParHeap, \WProp}
  \end{array}
\end{gather*}

\begin{notation}[Forcing for $\WProp$]\label[notation]{notation:world-kripke-joyal}
  For any $X:\SemLType$, $\phi:X\to \WProp$ and $w:\Worlds$ and $x:Xw$, we
  shall write $\Hl{\WForce{w}{\phi\,x}}$ in \iGDTT{} to mean that $\Hl{\phi_w x\, w}$
  holds.
\end{notation}

\begin{notation}[Forcing for $\BIAlg$]\label[notation]{notation:heap-kripke-joyal}
  Let $X:\SemLType$ be a semantic type; then in the internal language of
  $\SemLType$, for any $\phi : \BIAlg^X$, $h : \ParHeap$, and $x:X$, we shall
  write $\Hl{\HForce{h}{\phi\,{x}}}$ to mean that $\Hl{\phi\,x\,h}$ holds.
\end{notation}

We note that $\SemLType$ inherits~\cite{palombi-sterling:2023} from \iGDTT{}
a later modality $\Ltr$ defined pointwise; it follows that $\WProp$ is
closed under a later modality $\PropLtr : \WProp\to\WProp$ satisfying
$\WForce{w}{\PropLtr\prn{\phi\,x}} \LEquiv \Ltr\prn{\WForce{w}{\phi\,x}}$.

\subsubsection{Interpretation of judgmental structure}

We summarize the interpretation of the judgmental structure of \LMuRef{}
below:

\begin{enumerate}

\item Type contexts \JdgBox{\IsTpCtx{\Xi}} are interpreted as semantic
  logical types $\bbrk{\Xi} : \UniType$.

\item Element contexts \JdgBox{\IsCtx{\Xi}{\Gamma}} are interpreted as
  families $\bbrk{\Gamma} : \bbrk{\Xi}\to\SemLType$.

\item Types \JdgBox{\IsTp{\iota}{\Xi}{A}} are interpreted as families
  $\bbrk{A} : \bbrk{\Xi}\to\SemType{\iota}$.

\item Elements \JdgBox{\IsEl{\Xi}{\Gamma}{a}{A}} are interpreted as elements $\bbrk{a} : \Prod{\xi:\bbrk{\Xi}}\bbrk{\Gamma}\xi \to \bbrk{A}\xi$.

\item Propositions \JdgBox{\IsProp{\Xi}{\Gamma}{\phi}} are interpreted as
  predicates $\bbrk{\phi} : \Prod{\xi:\Xi}\BIHyp{\bbrk{\Gamma}\xi}$.

\item Entailments \JdgBox{\Entails{\Xi}{\Gamma}{\phi}{\psi}} are
  interpreted as parameterized inequalities $\Forall{\xi:\bbrk{\Xi}}
  \bbrk{\phi}\xi\leq\Sub{\BIHyp{\bbrk{\Gamma}\xi}} \bbrk{\psi}\xi$.

\end{enumerate}

\subsubsection{Recursive types, general reference types, and the monad}

In our semantics, recursive types are computed using the guarded domain
structure of the semantic universes $\SemType{\iota}$; general
reference types are defined pointwise as a subtype of the world's support; the
monad is likewise defined pointwise using a combination of universal types,
existential types, and the guarded lifting monad:
\begin{gather*}
  \begin{array}{l}
    \mu : \prn{\SemType{\iota}\to\SemType{\iota}}\to \SemType{\iota}\\
    \mu F = \Con{gfix}\,\prn{\vartheta\Sub{\SemType{\iota}} \circ \Ltr{F}}
  \end{array}
  \
  \begin{array}{l}
    \Con{ref} : \SemPType\to\SemPType\\
    \Con{ref}\,A\,w =
    \Compr{\ell : \DelimMin{1} \vrt{w}}{
    w \ell = \Next{A}
    }
  \end{array}
  \quad
  \begin{array}{l}
    \MonadT : \SemPType\to \SemPType\\
    \MonadT{A}w =
    \Forall{w'\geq w}
    \TotalHeap_{w'}
    \to
    \MonadL
    \Exists{w''\geq w'}
    \TotalHeap_{w''}
    \times
    Aw''
  \end{array}
\end{gather*}

We do not have the space to display the operations of the monad; we note,
however, that $\MonadT$ is $\SemPType$-enriched and therefore strong.  It
follows that the semantic type operations in this section
can be used to interpret the types of \FMuRef{}. We show how to interpret the
getter and setter for reference types in the model:
\begin{gather*}
  \begin{array}[t]{l}
    \Con{set}_A : \Con{ref}\,A\times A\to\MonadT\prn{}\\
    \prn{\Con{set}_A}_w\,\prn{\ell:\Con{ref}\,A\,w,a:A\,w}\,\prn{w'\geq w}\,\prn{h:\TotalHeap_{w'}} \triangleq\\
    \quad
    \eta\, \prn{
      \Con{pack}\, \prn{
        w', h\brk{\ell \mapsto \Con{next}\, w'_*a} , \prn{}
      }
    }
  \end{array}
  \qquad
  \begin{array}[t]{l}
    \Con{get}_A : \Con{ref}\,A\to\MonadT{A}\\
    \prn{\Con{get}_A}_w\,\prn{\ell:\Con{ref}\,A\,w}\,\prn{w'\geq w}\,\prn{h:\TotalHeap_{w'}} \triangleq\\
    \quad
    \Alert{\vartheta}\, \prn{
      \Next*\brk{B \gets w'\ell, x \gets h\ell}.\,
      \eta\, \prn{\Con{pack}\, \prn{w', h, x}}
    }
  \end{array}
\end{gather*}

Note that $\Con{set}_A$ returns immediately in the guarded lift monad via the
unit $\eta$, whereas $\Con{get}_A$ takes a single step via the $\Ltr$-algebra
map $\vartheta$; this is because the heap stores its elements under the later
modality, so reading from memory takes one abstract step of computation in the
guarded lift monad. This is also reflected in the rule \textsc{wp-get}, which allows
an assumption to be under the later modality.

\subsubsection{Semantics of logical types}

We interpret the logical type of propositions $\TpProp$ as the internal poset
$\BIAlg = \brk{\ParHeap,\WProp}$. The interpretation of the remaining type
connectives is standard.

\subsection{Semantics of predicate connectives}\label{sec:semantics-of-predicates}

We will impose enough structure on $\BIAlg$ such that the
indexed partial order $\BIHyp : \OpCat{\SemLType}\to
\mathbf{Poset}\Sub{\UniType}$ has the structure of a BI-hyperdoctrine with
appropriate modalities ($\square$, $\rhd$) and weakest preconditions.


\subsubsection{A complete BI-algebra}

We will argue that $\BIAlg$ forms a \DefEmph{complete
BI-algebra} in $\SemLType$. Note that in this section, when we say that a
partial order is \emph{complete}, we mean that it is complete in the sense of
internal category theory~\cite{jacobs:1999}.

\NewDocumentCommand\LemWPropCHABody{}{
  $\WProp$ is a complete Heyting algebra in $\SemLType$.
}

\begin{lemma}\label[lemma]{lem:wprop-cha}
  \LemWPropCHABody
\end{lemma}

\begin{corollary}
  The internal poset $\BIAlg$ is a complete
  Heyting algebra in $\SemLType$.
\end{corollary}

We will use \DefEmph{Day's convolution}~\cite{day:1970,day:1974} to
construct a BI-algebra structure on $\BIAlg = \brk{\ParHeap,\WProp}$.
Day's convolution product is most well-known for extending monoidal
structures on small categories to presheaves, but we will need the full
generality of his result: the categories involved are $\WProp$-enriched, and
the structure on the base is \emph{promonoidal} rather than \emph{monoidal}.

\begin{construction}[$\WProp$-enriched promonoidal structure on a pcm]\label[construction]{con:promonoidal}
  A partial commutative monoid $M = \prn{M,\emptyset,\cdot}$ in $\SemLType$ can be viewed as a $\WProp$-enriched category, because its extension order is valued in $\WProp$. As a $\WProp$-enriched category, $M$ has
  a $\WProp$-enriched \DefEmph{promonoidal structure} in the sense of
  Day~\cite{day:1970,day:1974}, which essentially encodes graph of the
  partial multiplication operation:
  \begin{enumerate}
    \item The $\WProp$-distributor $\Con{mul} : M\times \OpCat{\prn{M\times M}} \to \WProp$
      sends $\prn{m, \prn{n_0,n_1}}$ to $\exists n_2. m = n_0\cdot n_1\cdot n_2$.

    \item The $\WProp$-distributor $\Con{unit} : {M}\times\OpCat{\mathbf{1}} \to\WProp$ sends $\prn{m,*}$ to the
      proposition $\prn{m = \emptyset}$.

    \item The associativity and unit isomorphisms are defined using the
      associativity and unit laws for the partial multiplication operation.

  \end{enumerate}
\end{construction}

\begin{construction}[BI algebra]
  We obtain a BI-algebra structure $\prn{\SepConj,\SepImp}$ on $\BIAlg =
  \brk{\ParHeap,\WProp}$ by taking the Day convolution of the induced
  $\WProp$-enriched promonoidal structure (\cref{con:promonoidal}) on the
  partial commutative monoid $\ParHeap$, such that the separating conjunction
  extends the partial multiplication operation on representables. In
  particular, if $h\# h'$ are two disjoint partial heaps, then $\mathbf{y}h
  \SepConj \mathbf{y}h' = \mathbf{y}\prn{h\cdot h'}$.
\end{construction}

These constructions are explained in more detail for BI-algebras arising from
partial commutative monoids by Bizjak and
Birkedal~\cite{bizjak-birkedal:2018}.

\subsubsection{Modalities: persistence and later}

The persistence modality is interpreted as the map $\square : \BIAlg\to\BIAlg$
obtained by reindexing along the constant endomap $h\mapsto \emptyset$,
sending $\phi:\BIAlg$ to $h\mapsto \phi\, \emptyset$. The later modality
$\rhd : \BIAlg\to\BIAlg$ is given \emph{pointwise}.

\subsubsection{The points-to predicate}

We interpret the points-to predicate in the generic case
$\bbrk{\IsProp{\alpha}{\ell:\TpRef\,A,a:A}{\PointsTo{\ell}{a}}}$. In
particular, we must define for each program type $A:\SemPType$ a natural
transformation $\PointsTo{-}{-}:\Con{ref}\,A\times A\to \BIAlg$, which will
turn out to be representable by a singleton heap, \ie we may define
$\HForce{h}{\PointsTo{\ell}{a}} \LEquiv \brc{l\mapsto\Next{a}} \leq h$.

\subsubsection{Weakest preconditions}

Finally we must interpret the weakest precondition connective, which we do in
the generic case of
$\bbrk{\IsProp{\alpha}{\phi:\alpha\to\TpProp,u:\TpT{\alpha}}{\Lwp\,u\,\brc{x.\phi\,
x}}}$. This amounts to constructing for each semantic program type
$A:\SemPType$ a natural transformation $\Swp_A : \BIAlg^A\times \MonadT{A}\to
\BIAlg$. As the denotation of the state monad is defined world-by-world, so
must be the interpretation of $\HForce{h}{\Swp_A\,u\,\phi}$; to that end, we
give the forcing clause for $\WForce{w}\prn{\HForce{h}{\Swp_A\,u\,\phi}}$ in
the external \iGDTT{} language as follows, recalling the $\Downarrow$ notation
for the predicate lifting of $\MonadL$ from
\cref{cnstr:guarded-family-lifting}:

\iblock{
  \mhang{
    \WForce{w}\prn{
      \HForce{h}{\Swp_A\,u\,\phi}
    }
    \LEquiv
  }{
    \mrow{
      \forall\prn{w'\geq w}
      \,\prn{h_f:\ParHeap w'}
      \,\prn{h_t : \TotalHeap_{w'}}
      \,\prn{h_t = h_f\cdot w'_*h}.
    }
    \mhang{
      u\,w'\,h_t \Downarrow \lambda p.\
    }{
      \mrow{
        \exists\prn{w''\geq \vrt{h}}
        \,\prn{w_r = w''\cdot\vrt{h_f}}
        \,\prn{h' : \TotalHeap_{w''}}
        \,\prn{a:A w_r}.
      }
      \mrow{
        p = \Con{pack}\,\prn{w_r, \prn{w_r}_*h_f\cdot\prn{w_r}_*h', a}
      }
      \mrow{
        \mathrel{\land}
        \WForce{w_r}\prn{
          \HForce{\prn{w_r}_*h'}{
            \phi_{w_r}\,a = \top
          }
        }
      }
    }
  }
}

Although the definition is quite technical, the idea is simple. The denotation
of a monadic program is a guarded-recursive process taking a heap and
ultimately producing a return configuration at a larger world. In simple terms,
the weakest precondition of a predicate $\phi$ should quantify over all frames
for the starting heap and check that the process returns only configurations
satisfying $\phi$ without disturbing the frame.

\subsubsection{Explicit Kripke--Joyal translation}

We have given the interpretation of our logic in a mostly abstract--categorical
way; such an abstract presentation verifies all the ``logical'' rules of our
system, but explicit computations are needed in order to verify the rules for
weakest preconditions. In this section, we provide some tools to assist with
these explicit computations; in \cref{lem:kripke-joyal:1} we describe how to
interpret each of the main connectives of the logic as a transformer of
subobjects in the internal language of $\SemLType$.

\begin{computation}[Kripke--Joyal translation of the $\BIAlg$
  logic]\label[computation]{lem:kripke-joyal:1}
  The action of each connective on $\BIAlg=\brk{\ParHeap,\WProp}$ can be
  computed explicitly as a forcing clause in the Kripke--Joyal
  translation~\cite{maclane-moerdijk:1992}. We omit the forcing clauses for
  $\top,\bot,\land,\lor,\exists,\PropLtr$ because they are pointwise:
  \begin{align*}
    \HForce{h}{\phi\,x\Rightarrow\psi\,x} &\LEquiv
    \forall\prn{h'\geq h}.\,
    \HForce{h'}{\phi\,x} \Rightarrow \HForce{h'}{\psi\,x}
    \\
    \HForce{h}{\forall_Y \phi\prn{x,-}} &\LEquiv
    \forall\prn{h'\geq h}\,\prn{y:Y}.\,
    \HForce{h'}{\phi\,\prn{x,y}}
    \\
    \HForce{h}{\square\prn{\phi\,x}} &\LEquiv
    \HForce{\emptyset}{\phi\,x}
    \\
    \HForce{h}{\phi\,x\SepConj\psi\,x} &\LEquiv
    \exists\prn{h_1\cdot h_2 = h}.\,
    \HForce{h_1}{\phi\,x} \land
    \HForce{h_2}{\psi\,x}
    \\
    \HForce{h}{\phi\,x\SepImp\psi\,x} &\LEquiv
    \forall\prn{h'\# h}.\,
    \HForce{h'}{\phi\,x} \Rightarrow
    \HForce{h\cdot h'}{\psi\,x}
    \\
    \HForce{h}{\PointsTo{l}{a}} &\LEquiv
    \brc{\ell\mapsto \Next{a}}\leq h
  \end{align*}
\end{computation}

\begin{computation}[Kripke--Joyal translation of the $\WProp$ logic]\label[computation]{lem:kripke-joyal:2}
  The connectives on $\WProp$ can be further computed in terms of the ambient
  \iGDTT{} model by another layer of Kripke--Joyal forcing:
  \begin{align*}
    \WForce{w}{\phi\,x\Rightarrow\psi\,x} &\LEquiv
    \forall\prn{w'\geq w}.\,
    \WForce{w'}{\phi\,\prn{w'_*x}} \Rightarrow\WForce{w'}{\psi\,\prn{w'_*x}}
    \\
    \WForce{w}{\forall_Y\phi\prn{x,-}} &\LEquiv
    \forall\prn{w'\geq w}\,\prn{y:Yw'}.\,
    \WForce{w'}{\phi\,\prn{w'_*x,y}}
  \end{align*}
\end{computation}

\subsection{Soundness results}\label{sec:soundness}

The following results are stated internally to \iGDTT.

\begin{theorem}[Soundness]
  If \Hl{$\Entails \Xi \Gamma \phi \psi$} is derivable in \LMuRef{},
  then for any $\xi : \bbrk{\IsTpCtx \Xi}$ it holds that
  $\bbrk{\IsProp \Xi \Gamma \phi} \xi \leq \bbrk{\IsProp \Xi \Gamma \psi} \xi$ in
  $\BIHyp{\bbrk{\IsCtx \Xi \Gamma} \xi}$.
\end{theorem}

\begin{proof}
  Since $\BIHyp$ is a BI-hyperdoctrine, all the specified rules of higher-order
  separation logic are valid~\cite{biering-birkedal-torp-smith:2007}, and the
  rules for the modalities follow similarly. What remains is to verify rules
  for weakest preconditions; we show just one illustrative case in
  \cref{lem:soundness-wp-get}.
\end{proof}

\begin{lemma}\label[lemma]{lem:soundness-wp-get}
  The following weakest precondition law for $\TmGet$ is valid:
  \[
    \Entails{\alpha}{\phi:\alpha\to\TpProp, \ell:\TpRef\,\alpha, a:\alpha}{
      \PointsTo{\ell}{a}
      \SepConj
      \PropLtr\prn{
        \PointsTo{\ell}{a}
        \SepImp
        \phi\, a
      }
    }{
      \Lwp\,\prn{\TmGet\,\ell}\,\brc{x.\phi\, x}
    }
  \]
\end{lemma}

\begin{proof}
  We will follow the Kripke--Joyal unfolding of the separation logic from \cref{lem:kripke-joyal:1}.
  Fixing $A:\SemPType$, we must prove the following:
  \iblock{
    \mhang{
      \forall\,\prn{h_1\# h_2 : \ParHeap}\,
      \prn{\phi : A\to \BIAlg}\,
      \prn{\ell:\Con{ref}\,A}\,
      \prn{a:A}.
    }{
      \mrow{
        \prn{\brc{\ell\mapsto\Next{a}}\leq h_1}
        \land
        \PropLtr\prn{
          \forall\prn{h_3\# h_2}.\
          {\brc{\ell\mapsto a}\leq h_3}
          \Rightarrow
          \HForce{h_2\cdot h_3}{\phi\,a}
        }
      }
      \mrow{
        \mathrel{\Rightarrow}
        \HForce{h_1\cdot h_2}{\Swp_A\,\prn{\Con{get}_A\,\ell}\,\phi}
      }
    }
  }

  At this point, the only way to unfold further is to pass through a second
  Kripke--Joyal translation (\cref{lem:kripke-joyal:2}), where the indexing
  comes from $\Worlds$ rather than $\ParHeap$. What follows, therefore, will be
  in ambient \iGDTT{} language rather than the internal language of
  $\SemLType$; in particular, we assume the following:

  \iblock{
    \mhang{
      \prn{w:\Worlds}\,
      \prn{h_1\# h_2 : \ParHeap w}\,
      \prn{\phi : \mathcal{B}^Aw}\,
      \prn{\ell : \Con{ref}\,A\,w}\,
      \prn{a : A\,w}\
      \text{such that:}
    }{
      \mrow{H_1 : \brc{\ell\mapsto \Next{a}}\leq h_1}
      \mrow{
        H_2 :
        \Ltr\prn{
          \DelimMin{1}
          \forall \prn{w'\geq w}\,\prn{h_3\# w'_*h_2 : \ParHeap w'}.\,
          {
            \brc{\ell\mapsto \Next{w'_*a}} \leq h_3
          }
          \Rightarrow
          \WForce{w'}\prn{
            \HForce{w'_*h_2\cdot h_3}{\phi_{w'}\,\prn{w'_*a} = \top}
          }
        }
      }
    }
  }

  Our goal is to prove $\WForce{w}\prn{\HForce{h_1\cdot
  h_2}{\Swp_A\,\prn{\Con{get}_A\,\ell}\,\phi}}$; we fix a world $w'\geq w$, a
  frame $h_f:\ParHeap w'$ and a total heap $h_t : \TotalHeap_{w'}$ such that
  $h_t = h_f\cdot w'_*h_1\cdot w'_*h_2$ to prove the following:

  \iblock{
    \mhang{
      \Con{get}_A\,\ell\,w'\,h_t \Downarrow \lambda p.\
    }{
      \mrow{
        \exists\prn{w''\geq \vrt{h_1\cdot h_2}}
        \,\prn{w_r = w''\cdot\vrt{h_f}}
        \,\prn{h' : \TotalHeap_{w''}}
        \,\prn{a:A w_r}.
      }
      \mrow{
        p = \Con{pack}\,\prn{w_r, \prn{w_r}_*h_f\cdot\prn{w_r}_*h', a}
      }
      \mrow{
        \mathrel{\land}
        \WForce{w_r}\prn{
          \HForce{\prn{w_r}_*h'}{
            \phi_{w_r}\,a = \top
          }
        }
      }
    }
  }

  Our assumption $H_1$ guarantees that $\Con{get}_A\,\ell\,w'\,h_t$ is equal to
  $\vartheta\,\prn{\Next\prn{\eta\,\prn{\Con{pack}\,\prn{w',h_t,w'_*a}}}}$.
  Therefore, our goal reduces to the following by definition of the
  $\MonadL$--predicate lifting:

  \iblock{
    \mhang{
      \Ltr\exists\prn{w''\geq \vrt{h_1\cdot h_2}}
      \,\prn{w_r = w''\cdot\vrt{h_f}}
      \,\prn{h' : \TotalHeap_{w''}}
      \,\prn{x:A w_r}.
    }{
      \mrow{
        \Con{pack}\,\prn{w',h_t,w'_*a} = \Con{pack}\,\prn{w_r, \prn{w_r}_*h_f\cdot\prn{w_r}_*h', x}
      }
      \mrow{
        \mathrel{\land}
        \WForce{w_r}\prn{
          \HForce{\prn{w_r}_*h'}{
            \phi_{w_r}\,x = \top
          }
        }
      }
    }
  }

  Going under the later modality in the goal, we discharge the corresponding
  modality from $H_2$. Then we instantiate $w'' \triangleq
  \vrt{h_1}\cdot\vrt{h_2}$ and $w_r \triangleq w'$ and $h' \triangleq
  w''_*h_1\cdot w''_*h_2$ and $x\triangleq w'_*a$. Our remaining goal is
  $
    \WForce{w'}\prn{
      \HForce{w'_*h_1\cdot w'_*h_2}{
        \phi_{w'}\,\prn{w'_*a} = \top
      }
    }
  $, which follows by instantiating $H_2$ with $w' \triangleq w'$ and $h_3\triangleq w'_*h_1$.
\end{proof}

\NewDocumentCommand\CorConsistencyBody{}{
  It is not the case that $\top \leq \bot$ in $\BIHyp{\bbrk{\EmpCx}\xi}$ for
  any $\xi : \bbrk{\Xi}$.
}

\begin{corollary}[Consistency]\label[corollary]{cor:consistency}
  \CorConsistencyBody
\end{corollary}

\section{Conclusions; related and future work}

\subsection{Comparison with operationally-based program logics}

We have contributed a program logic \LMuRef{} over the equational theory of
polymorphic, general recursive, higher-order stateful programs by building on
the recent denotational semantics of general references and polymorphism of
Sterling~\etal~\cite{sterling-gratzer-birkedal:2022}, adapting many ideas that
were first developed in the context of operationally based program logics.
Prior works on the operational side such as Iris~\cite{iris:2018}, the Verified
Software Toolchain~\cite{appel:2011}, and
TaDA~\cite{de-rocha-pinto-dinsdale-young-garnder:2014} have reached great
heights of expressivity, incorporating constructs such as invariants and
higher-order ghost state which are critical for reasoning about concurrent
programs.
For the sake of simplicity, we have restricted our attention to a \emph{fixed}
notion of resource (partial heaps), but we hope in the future to adapt more
sophisticated constructs including higher-order ghost state, \etc to reach
parity with existing operationally-based program logics.

\subsection{Other denotationally-based program logics for state}

The model of Sterling~\etal is not the only denotational semantics of state.
First-order store, both local and global, is well-represented in the
literature~\cite{oles:1986,moggi:1991,plotkin-power:2002,staton:2010}; there is
also Levy's model of non-polymorphic higher-order store~\cite{levy:2002}, and
notably, a model of \emph{\textbf{local} full ground store} by
Kammar~\etal~\cite{kammar-levy-moss-staton:2017}.  Polzer and
Goncharov~\cite{polzer-goncharov:2020} have constructed a BI-hyperdoctrine over
the denotational semantics of Kammar~\etal, and our own work is much in the
spirit of theirs. However, there is an apparent mismatch between the semantics
of \emph{local store} and the model of bunched implications over it, which has
impeded \opcit from developing a full program logic with an interpretation of
weakest preconditions.


\subsection{Future perspectives}

One of the methodological questions raised by our work is where, exactly, to
draw the line between equational reasoning and logical reasoning. For instance,
conventional operationally-based program logics do not use equational reasoning
at all: our logic, in contrast, allows some equational reasoning but it is
limited by the intensionality of Sterling~\etal's model. One possible direction
for future work is to attempt to make the model itself less intensional, either
by enhancing the semantic worlds~\cite{dreyer-neis-rossberg-birkedal:2010} or
by improving the interpretation of the state
monad~\cite{kammar-levy-moss-staton:2017}.

Another question is how \LMuRef{} can be implemented in a practical tool. Currently, definitional and logical equality coincide due to equality reflection, which leads the former to become undecidable. An implementation will therefore require a more refined account of the interaction between definitional and logical equality.

\section*{Acknowledgments}

We wish to thank Daniel Gratzer for many helpful discussions.  This research
was supported in part by a Villum Investigator grant (no.~25804), Center for
Basic Research in Program Verification (CPV), from the VILLUM Foundation, and
in part by the European Union under the Marie Sk\l{}odowska-Curie Actions
Postdoctoral Fellowship project
\href{https://cordis.europa.eu/project/id/101065303}{\emph{TypeSynth: synthetic
methods in program verification}}. Views and opinions expressed are however
those of the authors only and do not necessarily reflect those of the European
Union or the European Commission. Neither the European Union nor the granting
authority can be held responsible for them.

\nocite{sterling-gratzer-birkedal:2022,benabou:2000:distributors}

\bibliographystyle{entics}

\clearpage

\appendix

\section{Omitted rules}\label{sec:omitted-rules}

\subsection{Equational theory of a strong monad}

\begin{mathpar}
  \inferrule{
    \IsEl{\Xi}{\Gamma}{u}{A}
    \\
    \IsEl{\Xi}{\Gamma,x:A}{v}{\TpT\,B}
  }{
    \EqEl{\Xi}{\Gamma}{
      x\leftarrow\TmRet\,u; v
    }{
      v\brk{u/x}
    }{
      \TpT\,B
    }
  }
  \and
  \inferrule{
    \IsEl{\Xi}{\Gamma}{u}{\TpT\,A}
  }{
    \EqEl{\Xi}{\Gamma}{x\leftarrow u; \TmRet\,x}{u}{\TpT\,A}
  }
  \and
  \inferrule{
    \IsEl{\Xi}{\Gamma}{u}{\TpT\,A}\\
    \IsEl{\Xi}{\Gamma,x:A}{v}{\TpT\,B}\\
    \IsEl{\Xi}{\Gamma,y:B}{v}{\TpT\,C}
  }{
    \EqEl{\Xi}{\Gamma}{
      y\leftarrow\prn{x\leftarrow u; v}; w
    }{
      x\leftarrow u; y\leftarrow v; w
    }{\TpT\,C}
  }
\end{mathpar}

\subsection{Equational theory of universal and existential types}

\begin{mathpar}
  \inferrule{
    \IsEl{\Xi,\alpha}{\Gamma}{u}{A}
    \\
    \IsPTp{\Xi}{B}
  }{
    \IsEl{\Xi}{\Gamma}{\prn{\lambda\alpha.u}B\equiv u\brk{B/\alpha}}{A\brk{B/\alpha}}
  }
  \and
  \inferrule{
    \IsEl{\Xi}{\Gamma}{u}{\forall\alpha.A}
  }{
    \IsEl{\Xi}{\Gamma}{u\equiv \Lambda\alpha. u\cdot\alpha}{\forall\alpha.A}
  }
  \and
  \inferrule{
    \IsPTp{\Xi}{B,C}\\
    \IsPTp{\Xi,\alpha}{A}\\
    \IsEl{\Xi}{\Gamma}{u}{A\brk{B/\alpha}}\\
    \IsEl{\Xi,\alpha}{\Gamma,x:A}{v}{C}
  }{
    \IsEl{\Xi}{\Gamma}{
      \mathtt{let}\ \mathtt{pack}\,\prn{\alpha,x} =
      \mathtt{pack}\,\prn{B,u}\ \mathtt{in}\ v
      \equiv
      v\brk{B,u/\alpha,x}
    }{
      C
    }
  }
  \and
  \inferrule{
    \IsPTp{\Xi}{C}\\
    \IsEl{\Xi}{\Gamma,z:\exists\alpha.A}{u}{C}\\
    \IsEl{\Xi}{\Gamma}{v}{\exists\alpha.A}
  }{
    \IsEl{\Xi}{\Gamma}{
      u\brk{v/z}\equiv
      \mathtt{let}\ \mathtt{pack}\,\prn{\alpha,x} = v\ \mathtt{in}\ u\brk{\mathtt{pack}\,\prn{\alpha,x}/z}
    }{C}
  }
\end{mathpar}

\section{Omitted proofs}

\RestateThm{lem:wprop-cha}{\LemWPropCHABody}

\begin{proof}
  The simplest way to see this is to embed $\SemLType=\brk{\Worlds,\UniType_0}$
  into the larger co-presheaf category $\ECat = \brk{\Worlds,\UniType_1}$,
  where we have a Hofmann--Streicher lifting $\VV$ of $\UniType_0$ whose global
  points correspond to $\SemLType$. Then the \emph{completeness} can be
  expressed internally in terms of quantification over $\VV$; this internal
  quantification automatically satisfies the appropriate Beck--Chevalley
  conditions when externalized. In particular, that $\WProp$ has internal
  products then amounts to the following pullback square
  existing~\cite{awodey:2018:natural-models}:
  \[
    \DiagramSquare{
      nw/style = pullback,
      east/style = >->,
      west/style = >->,
      north = !\Sub{\VV},
      south/style = {exists,->},
      ne = \mathbf{1},
      se = \WProp{},
      east = \top,
      sw = \Sum{A:\VV}\prn{A\to \WProp},
      nw = \VV,
      west = A\mapsto \prn{A,\lambda\_.\top},
      south = \prn{\forall},
      width = 4cm,
    }
  \]

  The lower map can be be constructed explicitly using the internal
  completeness of $\UniP$ in $\UniType$. The rest of the Heyting algebra
  structure is inherited from $\UniP$ \`a la Kripke semantics over $\Worlds$.
\end{proof}

\RestateThm{cor:consistency}{\CorConsistencyBody}

\begin{proof}
  The ordering on $\BIHyp{\bbrk{\EmpCx} \xi}$ is defined pointwise
  in relative to ordering on $\BIAlg$, and its ordering is in turn defined
  pointwise relative to the ordering on $\UniP$. Similarly, $\top$ and $\bot$
  in $\BIHyp{\bbrk{\EmpCx} \xi}$ are defined pointwise relative
  to $\top$ and $\bot$ in $\BIAlg$, which in turn are defined pointwise relative
  to $\top$ and $\bot$ in $\UniP$. Since $\bbrk{\EmpCx} \xi$ is the
  terminal object, and thus in particular is inhabited, this implies that if $\top
  \leq \bot$ in $\BIHyp{\bbrk{\EmpCx} \xi}$, we also have $\top
  \leq \bot$ in $\UniP$. We thus conclude $\top \not \leq \bot$ in $\BIHyp{\bbrk{\EmpCx} \xi}$.
\end{proof}

\end{document}